\documentclass[letterpaper, 10 pt, conference]{ieeeconf}  % Comment this line out if you need a4paper

\IEEEoverridecommandlockouts                              %for thanks command
\usepackage{verbatim}
\usepackage{graphics} % for pdf, bitmapped graphics files
\usepackage{amsmath,amssymb,stmaryrd,mathtools, amsfonts}
\usepackage{multirow}% http://ctan.org/pkg/multirow
\usepackage{hhline}
\usepackage{algorithm}
\usepackage{footmisc}
\usepackage[noend]{algpseudocode}
\usepackage{subcaption}
\PassOptionsToPackage{hyphens}{url}\usepackage{hyperref} % break url at hyphens
\usepackage{dsfont}
\usepackage{algorithmicx}
\usepackage{multirow}
\usepackage[utf8]{inputenc}
\usepackage{import}
\usepackage{graphicx, graphics}
\usepackage{caption}
\usepackage{subcaption}
\usepackage{array}
\usepackage{color}
\usepackage[table]{xcolor}
\usepackage{siunitx}
\usepackage[short]{optidef}
\usepackage{float}
\definecolor{myedit}{rgb}{1.0,0,0}
\definecolor{mytodo}{rgb}{0,1.0,1.0}
\definecolor{lightgray}{rgb}{0.8, 0.8, 0.8}
\newcommand{\todo}[1]{{\color{mytodo} #1}} % for things left to change

\newtheorem{prop}{Proposition}
\newcommand{\M}[1]{\underset{#1}{\text{minimize}}\;}

\newcommand{\mobility}[0]{$\tilde{\mathcal{T}}_{episode}$}
\newcommand{\conservatism}[0]
                    {$\Delta\tau $ (m)      
                    & $\bar d_{min}$ (m)
                    & $\mathcal{F}$ (\%)}
\newcommand{\comfort}[0]{$\tilde{\mathcal{A}}_{lat}$ 
                        & $\bar{\mathcal{J}}_{long}$ ($\frac{m}{s^3}$)
                        & $\bar{\mathcal{J}}_{lat}$ ($\frac{m}{s^3}$)
                        }
\newcommand{\efficiency}[0]{$\bar{\mathcal{T}}_{solve}$ (ms)}

\title{\LARGE \bf 
Stochastic MPC with Multi-modal Predictions for Traffic Intersections
}

\author{Siddharth H. Nair$^{\star}$, Vijay Govindarajan$^{\star}$, Theresa Lin, Chris Meissen, H. Eric Tseng, Francesco Borrelli
\thanks{$^\star{}$Indicates equal contribution.
SHN, VG and FB are with the Model Predictive Control Laboratory, UC Berkeley. 
TL and CM are with Ford AV LLC.
HET is with Ford Research and Advanced Engineering.
\{siddharth\_nair, govvijay, fborrelli\}@berkeley.edu and \{tlin33, cmeissen, htseng\}@ford.com.}
}

\begin{document}

\maketitle
\thispagestyle{empty}
\pagestyle{empty}

%%%%%%%%%%%%%%%%%%%%%%%%%%%%%%%%%%%%%%%%%%%%%%%%%%%%%%%%%%%%%%%%%%%%%%%%%%%%%%%%
\begin{abstract}
We propose a Stochastic MPC (SMPC) formulation for autonomous driving at traffic intersections which incorporates multi-modal predictions of surrounding vehicles given by Gaussian Mixture Models (GMM) for collision avoidance constraints. 
Our main theoretical contribution is a SMPC formulation that optimizes over a novel feedback policy class designed to exploit additional structure in the GMM predictions, and that is amenable to convex programming. The use of feedback policies for prediction is motivated by the need for reduced conservatism in handling multi-modal predictions of the surrounding vehicles, especially prevalent in traffic intersection scenarios. We evaluate our algorithm along axes of mobility, comfort, conservatism and computational efficiency at a simulated intersection in CARLA.
Our simulations use a kinematic bicycle model and multimodal predictions  trained on  a subset of the Lyft Level 5 prediction dataset.
To demonstrate the impact of optimizing over feedback policies, we compare our algorithm with two SMPC baselines that handle multi-modal collision avoidance chance constraints by optimizing over open-loop sequences. 
\end{abstract}

%%%%%%%%%%%%%%%%%%%%%%%%%%%%%%%%%%%%%%%%%%%%%%%%%%%%%%%%%%%%%%%%%%%%%%%%%%%%%%%%
\section{Introduction}
\label{sec:introduction}
\subsection*{Motivation}
Autonomous vehicle technologies have seen a surge in popularity over the last decade, with the potential to improve flow of traffic, safety and fuel efficiency \cite{nhtsa}. While existing technology is being gradually introduced into scenarios such as highway driving and low-speed parking, the traffic intersection scenario presents a complex challenge, especially in the absence of V2V communication.  In any traffic scenario, an autonomous agent must plan to follow its desired route while accounting for surrounding agents to maintain safety. The difficulty arises from the variability in the possible behaviors of the surrounding agents \cite{nhtsa_intersection}, \cite{wei2021autonomous}. To address this difficulty, significant research has been devoted to modelling these agent predictions as multi-modal distributions \cite{imm_bar_shalom_1988, multipath_2019, trajectron_2020}. Such models capture uncertainty in both high-level decisions (desired route) and low-level executions (agent position, heading, speed).

The focus of this work is to incorporate these multi-modal distributions for the surrounding agents (target vehicles) into a planning framework for the autonomous agent (ego vehicle). The main challenge in designing such a framework is to find a good balance between performance, safety, and computation time.  We investigate this planning problem in the context of constrained optimal control and use Model Predictive Control (MPC), the state-of-the-art for real-time optimal control \cite{morari1999model, benblog}. We assume that predictions of the other agents are specified as Gaussian Mixture Models (GMMs) and propose a Stochastic Model Predictive Control (SMPC) framework that incorporates probabilistic collision avoidance and actuation constraints. 

\subsection*{Related work}
There is a large body of work focusing on the application of SMPC to autonomous driving \cite{planning_and_dm_for_avs_2018, rosolia2018data}. A typical SMPC algorithm involves solving a chance-constrained finite horizon optimal control problem in a receding horizon fashion \cite{mesbah2016stochastic}. The chance-constrained SMPC optimization problem, %while usually difficult to exactly reformulate except for special distributions,
offers a less conservative modelling framework for handling constraints in uncertain environments compared to Robust MPC frameworks. In the context of autonomous driving, SMPC has been used for imposing chance constraints accounting for uncertainty in vehicles' predictions in applications such as autonomous lane change \cite{carvalho2014stochastic, gray2013stochastic}, cruise control \cite{moser2017flexible} and platooning \cite{causevic2020information}. A common feature of works in these applications is that the predictions are either uni-modal or the underlying mode of the multi-modal prediction can be inferred accurately. 

In order to handle multi-modal predictions (specifically GMMs), the work in \cite{zhou2018joint} and \cite{wang2020non} proposes nonlinear SMPC algorithms that suitably reformulate the collision avoidance chance constraint for all possible modes and demonstrate their algorithms at traffic intersection scenarios. However, a non-convex optimization problem is formulated to find a single open-loop input sequence that satisfies the collision avoidance chance constraints for all modes and possible evolutions of the target vehicles over the prediction horizon given by the GMM. This approach can be conservative and a feasible solution to the optimization problem may not exist. To remedy this issue, we optimize over a class of feedback policies \cite{goulart2006optimization} which adds flexibility due to the ability to react to different realizations of the vehicles' trajectory predictions along the prediction horizon. 

The work in~\cite{schildbach2016collision} and \cite{batkovic2020robust} is the closest to our approach. The former considers uni-modal predictions and fixes a policy (computed offline) for which the optimization problem can still be conservative. The latter proposes a Robust MPC scheme that optimizes over mode-dependent sequences for multi-modal distributions of the obstacle with bounded polytopic supports. However, the sequences are not a function of vehicles' state trajectory predictions because this would require enumerating over all possible sequences of vertices of the support of the obstacle distribution, resulting in a formulation where the number of constraints is exponential in the prediction horizon.

\subsection*{Contributions}
Our main contributions are summarised as follows:
\begin{itemize}
    \item We propose a SMPC framework that optimizes over a novel class of feedback policies designed to exploit additional structure in the GMM predictions. These policies assume feedback over the ego vehicle state and the target vehicle positions, thus reducing conservatism. Moreover, we show that our SMPC optimization problem can be posed as a Second-Order Cone Program (SOCP).
    
    \item We present a systematic evaluation of our framework along axes of (i) mobility, (ii) comfort, (iii) conservatism and (iv) computational efficiency at a simulated traffic intersection. To demonstrate the impact of optimizing over feedback policies, we include two baselines: (i) a chance-constrained Nonlinear MPC optimizing over open-loop sequences with the multi-modal collision avoidance chance-constraints as in \cite{wang2020non} and (ii) an ablation of our SMPC formulation, which optimizes over open-loop sequences instead of feedback policies.
\end{itemize}
In the interest of space, we defer proofs and extensive simulation results to our website\footnote[1]{Supplementary material, videos, code @ \todo{\href{https://sites.google.com/view/siddharthnair/research/ICRA2022}{Project Website} }\label{website}}.

% \input{Sections/proof}
% \newpage
\section{Problem Formulation}
\label{sec:prblm_f}
In this section we formally cast the problem of designing SMPC in the context of autonomous driving at intersections.

\subsection{Preliminaries}

\subsubsection*{Notation} 
The index set $\{k_1,k_1+1,\dots, k_2\}$ is denoted by $\mathcal{I}_{k_1}^{k_2}$. For any two positive semi-definite matrices $M_1, M_2\in\mathbb{S}^n_+$, we have $M_1\prec M_2\Leftrightarrow M_2-M_1\in\mathbb{S}^n_+$. We denote $\Vert\cdot\Vert$ by the Euclidean norm and $\Vert x\Vert_{M}=\Vert \sqrt{M}x\Vert$ for some $M\succ 0$. 
Given random variables $X$ and $Y$, we write $X|Y$ as a shorthand for $X|Y=y$ ($X$ conditioned on $Y=y$).  %, abusing notation to denote the both the random variable and a particular realization by $Y$. 
Given a normal distribution $\mathcal{N}(\mu,\Sigma)$ and $\beta\geq 0$, define the $\beta-$confidence ellipsoid $\mathcal{E}(\mu, \Sigma, \beta)=\{x: \Vert x-\mu\Vert^2_{\Sigma^{-1}}\leq\beta\}$.  
The binary operator $\otimes$ denotes the Kronecker product.   
The partial derivative of function $f(x,u)$ with respect to $x$ at $(x,u)=(x_0,u_0)$ is denoted by $\partial_x f(x_0,u_0)$. %, \small$\dfrac{\partial}{\partial x}f(x,u)\Big\vert_{(x_0,u_0)}$\normalsize, is denoted by the shorthand \small$\partial_x f(x_0,u_0)$\normalsize.

\begin{figure}[h]
    \centering
    \includegraphics[width=0.35\columnwidth]{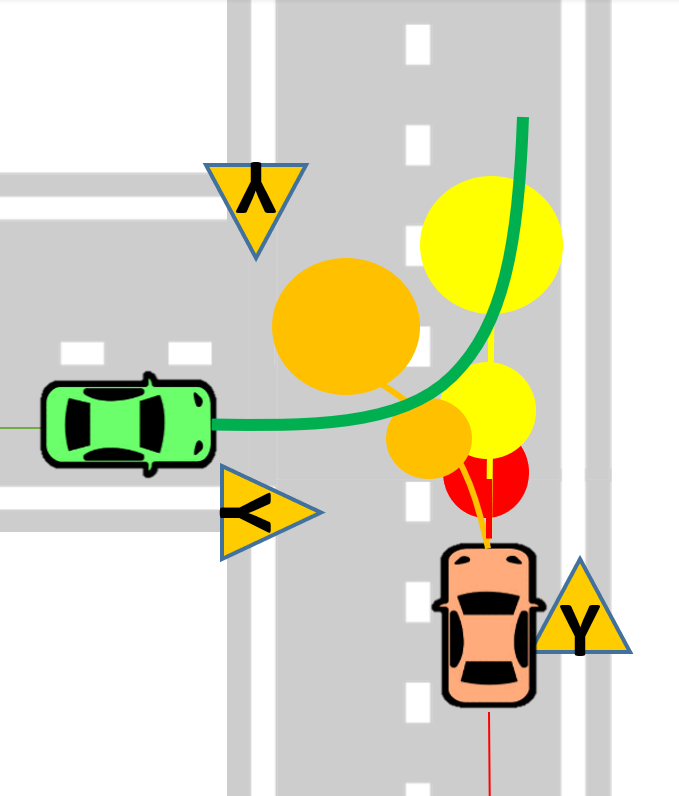}
        \caption{\small{Depiction of an unsignalized intersection with green Ego Vehicle (EV) and orange Target Vehicle (TV). The EV arrives first and intends to turn left (green path). The EV is provided 3 possible predictions of TV behavior in decreasing order of probability (red: give EV right-of-way, orange: turn left and yellow: go straight).}}
%    \caption{\small{Depiction of an unsignalized intersection with green Ego Vehicle (EV) and orange Target Vehicle (TV). The EV arrives at the intersection first and intends to turn left (green path). The EV is provided 3 possible predictions of TV behaviour in decreasing order of probability (red: give EV right-of-way, orange: turn left and yellow: go straight).}}
    \label{fig:intersection}    
\end{figure}
\subsubsection*{Intersection setup}
Consider the intersection shown in Figure~\ref{fig:intersection}.  The Ego Vehicle (EV), depicted in green, is the vehicle to be controlled.  All other vehicles at the intersection are called Target Vehicles (TVs).
%and depicted in orange (one shown in Figure~\ref{fig:intersection}).

\subsubsection*{EV modeling}
Let $x_t=[X_t\ Y_t\ \psi_t\ v_t]^\top$ be the state of the EV at time $t$ with $(X_t,Y_t)$ and $\psi_t$ being its position and heading respectively in global coordinates, and $v_t$ being its speed. The discrete-time dynamics of the EV are given by the Euler discretization of the kinematic bicycle model~\cite{kong2015kinematic}, denoted as $x_{t+1}=f^{EV}(x_t,u_t)$.
% \small
% \begin{align}\label{eq:EV_model}
%     x_{t+1}&=f^{EV}(x_t,u_t).
% \end{align}\normalsize
% \small
% \begin{align}\label{eq:EV_model}
%     \begin{bmatrix}X_{t+1}\\Y_{t+1}\\\theta_{t+1}\\v_{t+1}\end{bmatrix}&=\begin{bmatrix}X_{t}\\Y_{t}\\\theta_{t}\\v_{t}\end{bmatrix}+dt\begin{bmatrix}v_t\cos(\theta_t+\beta_t)\\v_t\sin(\theta_t+\beta_t)\\\dfrac{v_t}{L_R}\sin(\theta_t+\beta_t)\\
%     a_t
%     \end{bmatrix}\nonumber\\
%     \beta_t&=\arctan\left(\dfrac{L_R}{L_F+L_R}\tan\delta_t\right)
% \end{align}\normalsize
The control inputs $u_t=[a_t\ \delta_t]^\top$ are the acceleration $a_t$ and front steering angle $\delta_t$. The system state and input constraints are given by polytopic sets $\mathcal{X}, \mathcal{U}$ respectively.

As depicted by the green path in Figure~\ref{fig:intersection}, we assume a kinematically feasible reference trajectory, \small\begin{align}\label{eq:EV_ref_traj}
    \{(x^{ref}_t,u^{ref}_{t})\}_{t=0}^{T}
\end{align}\normalsize is provided for the EV. This serves as the EV's desired trajectory which can be computed offline (or online at low frequency) by solving a nonlinear trajectory optimization problem, while accounting for the EV's route, actuation limits, and static environment constraints (like lane boundaries, traffic rules).  However, this reference does not consider the dynamically evolving TVs for real-time obstacle avoidance.

\subsubsection*{TV predictions}
We focus our presentation on a single TV although the proposed framework can be generalised for multiple TVs. Let $o_t=[X^o_t\ Y^o_t]^\top$ be the position (and state) of the TV at time $t$. For the purpose of dynamic obstacle avoidance, we assume that we are provided predictions of the TV's position for $N$ time steps in the future, given by the random variables $\{o_{k|t}\}_{k=1}^{N}$, where each random variable $o_{k|t}$ is shorthand for $o_{t+k}|o_t$, the position of the TV at time $t+k$ conditioned on the current position $o_t$.  
The predictions for the TV at time $t$ is encoded as a Gaussian Mixture Model (GMM):\small
\begin{equation}\label{eq:GMM}
    \mathcal{G}_t= \left\{ p_{t,j}, \{\mathcal{N}(\mu_{k|t,j}, \Sigma_{k|t,j})\}_{k=1}^{N} \right\}_{j=1}^J
\end{equation}\normalsize
where
\begin{itemize}
    % \itemizespacing
    \item $J$ is the number of possible modes.%, and so the actual mode $\sigma_t\in\mathcal{I}_1^J$ ($J=3$ in Figure~\ref{fig:intersection}, coloured red, yellow and orange).
    \item $\mathbb{P}(\sigma_t=j|o_t)=p_{t,j} \in [0, 1]$ is the probability that the mode at time $t$, denoted by $\sigma_t$, is $j$ for some $j\in\mathcal{I}_1^J$.%is the probability of the $j$th mode for $j\in\mathcal{I}_1^J$ (darker colours depict modes with higher probability in Figure~\ref{fig:intersection}).
    \item $o_{k|t,j}\sim\mathcal{N}(\mu_{k|t,j}, \Sigma_{k|t,j})$, i.e., the position of the TV at time $t+k$ conditioned on the current position $o_t$ and mode $\sigma_t=j$, denoted by $o_{k|t,j}$, is given by the Gaussian distribution $\mathcal{N}(\mu_{k|t,j}, \Sigma_{k|t,j})$.
    %(depicted in Figure~\ref{fig:intersection} as a $\beta-$confidence ellipsoidal set $\mathcal{E}(\mu_{k|t,j},\Sigma_{k|t,j},\beta)$).    
\end{itemize}
The multi-modal distribution of $o_{k|t}$ is thus given by \small\begin{align}\label{eq:mm_dist_GMM}o_{k|t}\sim\sum_{j=1}^Jp_{t,j}\mathcal{N}(\mu_{k|t,j},\Sigma_{k|t,j})~~\forall k\in\mathcal{I}_1^{N}.\end{align}\normalsize

%Such predictions have been proposed in the literature specifically for autonomous driving applications and in this work, we use our implementation of MultiPath \cite{multipath_2019} to produce the GMM \eqref{eq:GMM}.
% \edit{This assumption allows us to handle a variety of model-based and data-driven GMM-based predictions.}
%% VG COMMENT: moving the rotation matrix assumption to the section where it's used (SMPC Obstacle Avoidance Constraints).
% We also assume that we are given predictions of the target vehicle heading as rotation matrices for each mode along the prediction horizon as 
%  \small\begin{align}\label{eq:TV_poses}
%      \{\{R_{k|t,j}\}_{k=1}^N\}_{j=1}^J.
%  \end{align}\normalsize
%  In our work, we approximately construct these rotation matrices using the GMM \eqref{eq:GMM} by defining the heading for mode $j$ at time $t+k$ as $\theta^o_{k|t}=\arctan\frac{[0\ 1](\mu_{k+1|t,j}-\mu_{k|t,j})}{[1\ 0](\mu_{k+1|t,j}-\mu_{k|t,j})}$.

 \subsection{Qualitative Control Objectives}
 \begin{itemize}
     \item Design a feedback control $u_t=\pi_t(x_t, o_t)$ for the EV to track the reference trajectory \eqref{eq:EV_ref_traj} while avoiding collisions with the TV, whose predictions are given by the GMM \eqref{eq:GMM}.%, \eqref{eq:TV_poses}. 
     \item The algorithm to compute $u_t$ must be computationally tractable and not overly conservative.
 \end{itemize}
 
 \subsection{Overview of Proposed Approach}
 \label{sec:smpc_approach}
 We propose a novel SMPC formulation to compute the feedback control $u_t$. The optimization problem of our SMPC takes the form,

 \begin{mini!}[2]
 {\substack{\theta_t}}{J_t(\mathbf{x}_t,\mathbf{u}_t)\label{opt:obj}}{\label{opt:SMPC_skeleton}}{}
\addConstraint{x_{k+1|t}}{=f^{EV}_k(x_{k|t},u_{k|t})\label{opt:EV_dyn}}
\addConstraint{o_{k+1|t}\vert o_{k|t}}{\sim f^{TV}_{k}(o_{k|t})\label{opt:TV_dyn}}
\addConstraint{\mathbb{P}(g_k(o_{k+1|t},x_{k+1|t})\leq0)}{\leq \epsilon\label{opt:oa_constr}}{}
\addConstraint{(x_{k+1|t},u_{k|t})\in\mathcal{X}\times\mathcal{U}}{\label{opt:ev_constr}}{}
\addConstraint{\mathbf{u}_t}{\in\Pi_{\theta_t}(\mathbf{x}_t,\mathbf{o}_t)\label{opt:gen_pol_class}}
\addConstraint{x_{0|t}=x_t,\ o_{0|t}=o_t }{\label{opt:init}}{}
\addConstraint{~ \forall k\in\mathcal{I}_0^{N-1}}{\nonumber}
\end{mini!}
where $\mathbf{u}_t=[u_{0|t},\dots, u_{N-1|t}]$, $\mathbf{x}_t=[x_{0|t},\dots, x_{N|t}]$ and $\mathbf{o}_t=[o_{0|t},\dots, o_{N|t}]$. 
The SMPC feedback control action is given by the optimal solution of \eqref{opt:SMPC_skeleton} as 
\begin{align}\label{eq:SMPC}u_t=\pi_{\mathrm{SMPC}}(x_t,o_t)=u^\star_{0|t}\end{align} where the EV and TV state feedback enters as \eqref{opt:init}.  

The objective \eqref{opt:obj} penalizes deviation of the EV trajectory from the reference \eqref{eq:EV_ref_traj} and the collision avoidance constraints are imposed as chance constraints \eqref{opt:oa_constr} along with polytopic state and input constraints $\mathcal{X},\mathcal{U}$ for the EV. The TV predictions in \eqref{opt:TV_dyn} are obtained from the GMM \eqref{eq:GMM} and assuming additional dynamical structure (discussed in Sec.~\ref{sec:smpc_tv_pred_model}). Using this, our SMPC optimizes over a novel parameterized policy class \eqref{opt:gen_pol_class} that depends on predictions of the EV's and TV's states as opposed to optimization over open-loop sequences which can often be conservative. This is our main theoretical contribution. Moreover, we show that \eqref{opt:SMPC_skeleton} can be posed as a convex optimization problem, making it computationally tractable and amenable to state-of-the-art solvers. This is achieved by (i) a convex parameterisation of policy class $\Pi_{\theta_t}(\cdot)$, (ii) using affine time varying models to generate EV and TV state predictions in \eqref{opt:EV_dyn} and \eqref{opt:TV_dyn}, (iii) using a convex cost function for \eqref{opt:obj}, and (iv) constructing convex inner approximations of \eqref{opt:oa_constr}. %This is discussed in Section~\ref{sec:SMPC}.
 
 %We incorporate our SMPC algorithm into our prediction and planning architecture for experiments at a simulated intersection in CARLA, as detailed in Section~\ref{sec:expt_design}. We evaluate our algorithm against various baselines to demonstrate the impact of optimizing over feedback policies and discuss our results in Section~\ref{sec:results}.

\section{Stochastic MPC with Multi-Modal Predictions}
\label{sec:SMPC}

In this section, we detail our SMPC formulation for the EV to track the reference \eqref{eq:EV_ref_traj} while incorporating multi-modal predictions from \eqref{eq:GMM} of the TV for obstacle avoidance. 
% \edit{We present propositions to justify our SMPC design; full details can be found in the extended supplement~\footnote{link here}}.

\subsection{EV prediction model}
The EV prediction model \eqref{opt:EV_dyn} is a linear time-varying model, obtained by linearizing $f^{EV}(\cdot)$ about the reference trajectory \eqref{eq:EV_ref_traj}. Defining $x_{k|t}-x^{ref}_{t+k}=\Delta x_{k|t}$ and $u_{k|t}-u^{ref}_{t+k}=\Delta u_{k|t}$, we have $\forall k\in\mathcal{I}_0^{N-1}$,
\small
\begin{align}\label{eq:EV_ltv_model}
    &\Delta x_{k+1|t}=A_{k|t}\Delta x_{k|t}+B_{k|t}\Delta u_{k|t}+w_{k|t}\\
    &A_{k|t}=\partial_xf^{EV}(x^{ref}_{t+k},u^{ref}_{t+k}),\ B_{k|t}=\partial_u f^{EV}(x^{ref}_{t+k},u^{ref}_{t+k})\nonumber
\end{align}
\normalsize
where the additive process noise $w_{k|t}\sim\mathcal{N}(0,\Sigma_w)$ (i.i.d with respect to $k$) models linearization error and other stochastic noise sources. Consequently, the polytopic state and input constraints \eqref{opt:ev_constr} are then written as chance constraints $\forall k\in\mathcal{I}_0^{N-1}$,
\small
\begin{align}\label{eq:EV_cc}
    &\mathbb{P}((\Delta x_{k+1|t},\Delta u_{k|t})\in\Delta\mathcal{X}_k\times\Delta\mathcal{U}_k )\geq 1-\epsilon\\
   &\Delta\mathcal{X}_k=\{\Delta x:F^x_k\Delta x\leq f^x_k \}, \Delta\mathcal{U}_k=\{\Delta u:F^u_k\Delta u\leq f^u_k \}. \nonumber 
\end{align}
\normalsize
\subsection{TV prediction model}\label{sec:smpc_tv_pred_model}
Notice that the multi-modal distribution of the random variable $o_{k+1|t}$ (as given by \eqref{eq:mm_dist_GMM}) is independent of the random variable $o_{k|t}$, i.e., $o_{k+1|t}\vert o_{k|t}=o_{k+1|t}$.
 However, since the TV positions are governed by the laws of motion, it is reasonable to assume dynamical structure to link positions at consecutive time steps such that $o_{k+1|t}\vert o_{k|t}\neq o_{k+1|t}$. Thus, a measurement of  $o_{k|t}$ reduces the uncertainty in $o_{k+1|t}$, which can be exploited by the policy class \eqref{opt:gen_pol_class} to satisfy \eqref{opt:oa_constr}. We construct a multi-modal distribution to model this assumed dynamical structure using the parameters of the GMM \eqref{eq:GMM} as follows. 
 
 We model dynamical behavior by assuming $~\forall k\in\mathcal{I}_1^{N-1}, \forall j\in\mathcal{I}_1^J$ that  if $o_{k|t,j}$ is $S$ standard deviations away from  $\mu_{k|t,j}$, then $o_{k+1|t,j}$ will also be approx. $S$ standard deviations away from  $\mu_{k+1|t,j}$. For safety, we would also like to ensure that the $\beta-$confidence sets for the distribution of $o_{k|t,j}$ given by the original GMM \eqref{eq:GMM} is contained within that of the new distribution.
 
 Consider the following candidate  for the conditional distribution $~\forall k\in\mathcal{I}_0^{N-1}, \forall j\in\mathcal{I}_1^J$:
%  $o_{k+1|t,j}$ conditioned on $o_{k|t,j}$ is given by the distribution 
%and mode $\sigma_t=j$ is distributed as
\small
\begin{align}\label{eq:tv_dist_j}
    o_{k+1|t,j}|o_{k|t,j}\sim\mathcal{N}(T_{k|t,j}o_{k|t,j}+c_{k|t,j},\Tilde{\Sigma}_{k+1|t,j}),
\end{align}
\normalsize
% \small
% \begin{align}\label{eq:mm_dist_markov}
%     o_{k+1|t}|o_{k|t}\sim\sum_{j=1}^Jp_{k|t,j}\mathcal{N}(T_{k|t,j}o_{k|t}+c_{k|t,j},\Tilde{\Sigma}_{k+1|t,j}),
% \end{align}
% \normalsize
where $\forall k\in\mathcal{I}_1^{N-1}$,\small$  T_{k|t,j}=\sqrt{\Sigma_{k+1|t,j}}\sqrt{\Sigma_{k|t,j}^{-1}}$, $c_{k|t,j}=\mu_{k+1|t,j}-T_{k|t,j}\mu_{k|t,j}$, $\Tilde{\Sigma}_{k+1|t,j}= \Sigma_n$ \normalsize for some \small$\Sigma_n\succ 0$, $\Sigma_n\prec\mathrel{\mkern-5mu}\prec \Sigma_{k+1|t,j}$\normalsize
% \small
% \begin{align}\label{eq:TV_atv}
%     T_{k|t,j}&=\sqrt{\Sigma_{k+1|t,j}}\sqrt{\Sigma_{k|t,j}^{-1}}, \ c_{k|t,j}=\mu_{k+1|t,j}-T_{k|t,j}\mu_{k|t,j},\nonumber\\
%     \Tilde{\Sigma}_{k+1|t,j}&= \Sigma_n\quad(\text{for some }\Sigma_n\prec\prec \Sigma_{k+1|t,j}, \Sigma_n\succ 0)
% \end{align}\normalsize
and \small$T_{0|t,j}=I,\ c_{0|t,j}=\mu_{1|t,j}-o_t,\ \tilde{\Sigma}_{1|t,j}=\Sigma_{1|t,j}$\normalsize.
In the next proposition, we compute the distribution of $o_{k|t,j}$ using \eqref{eq:tv_dist_j} and verify our modelling assumptions.
\begin{prop}\label{prop:1} Given the conditional distribution \eqref{eq:tv_dist_j}, we have $\forall j\in\mathcal{I}_1^J$, $\forall k\in\mathcal{I}_1^{N}$: 
\begin{itemize}
    \item  if \small$\Vert o_{k|t,j}-\mu_{k|t,j}\Vert^2_{\Sigma^{-1}_{k|t,j}}=S^2$\normalsize\ then \small$\mathbb{E}[\Vert o_{k+1|t,j}-\mu_{k+1|t,j}\Vert^2_{\Sigma^{-1}_{k+1|t,j}}|\ o_{k|t,j}]=S^2+\text{tr}(\Sigma^{-1}_{k+1|t,j}\Sigma_n)$ \normalsize .
    \item 
\small $o_{k|t,j}\sim\mathcal{N}(\mu_{k|t,j},\Sigma_{k|t,j}+\Sigma_n+\bar{\Sigma}_{k|t,j})$\normalsize\
where \footnotesize
$\bar{\Sigma}_{k|t,j}=\sum\limits_{m=1}^{k-2}\sqrt{\Sigma_{k-1|t,j}\Sigma_{m|t,j}^{-1}}\Sigma_{n}\sqrt{\Sigma_{k-1|t,j}\Sigma_{m|t,j}^{-1}}^\top$\normalsize. So the $\beta-$confidence ellipsoid of the new distribution contains that of the distribution given by the GMM in \eqref{eq:mm_dist_GMM}, i.e., \small$\mathcal{E}(\mu_{k|t,j},\Sigma_{k|t,j}+\Sigma_n+\bar{\Sigma}_{k|t,j}, \beta)\supset \mathcal{E}(\mu_{k|t,j},\Sigma_{k|t,j}, \beta)$\normalsize.
\end{itemize}
\end{prop}
\begin{proof}
Appendix~\ref{app:prop1proof} in full version \cite{nair2021stochastic}.
\end{proof}
 
% (note that $\Sigma_n+\bar\Sigma_{k|t,j}\approx O$ if $\Sigma_n\prec\prec \Sigma_{k|t,j}~ \forall k\in\mathcal{I}_1^{N}$).

Model \eqref{eq:tv_dist_j} describes inference over the state of the TV by conditioning the observations of the TV states on a given mode $\sigma_t=j$. To complete the description of the TV's multi-modal predictions, we now discuss inference over the mode $\sigma_t$ using observations $o_{k|t}$. There are sophisticated methods for computing the posterior distribution $\mathbb{P}(\sigma_t=j|o_{k|t})=p_{k|t,j}$ (e.g. using Bayes' rule or E-M) but these would complicate our SMPC problem \eqref{opt:SMPC_skeleton}. 
Instead, we assume $p_{k|t,j}=p_{t,j}$ $\forall k<\bar k$ and at some $\bar k$,  the mode can be exactly inferred from the TV's state. Motivated by \cite{batkovic2020robust}, our model for determining $\bar k$ is given by  
\small
\begin{align}\label{eq:k_bar}
    \min\big\{k\in\mathcal{I}_1^N:\ \mathcal{E}(\mu_{l|t,j_1},&\Sigma_{l|t,j_1},\beta)\cap\mathcal{E}(\mu_{l|t,j_2},\Sigma_{l|t,j_2},\beta)=\emptyset, \nonumber\\
    &\forall l\geq k, \forall j_1,j_2\in\mathcal{I}_1^J\big\}
\end{align}
\normalsize
Thus $\bar k$ is the minimum time step along the prediction horizon such that all $\beta-$confidence ellipsoids for subsequent time steps are pairwise disjoint for all modes. Given $\bar k$, we assume that the mode of the $i$th TV can be determined using state $o_{\bar k|t}$ as given by
\small
\begin{align}\label{eq:belief_update}
   p_{\bar k|t,j}=\begin{cases}
    1 & \text{ if } o_{\bar k|t}\in\mathcal{E}(\mu_{\bar{k}|t,j}, \Sigma_{\bar{k}|t,j}, \beta)\\ %\footnote{Alternatively, if $j=\text{arg}\min\limits_{j\in\mathcal{I}_1^J}(o_{\bar k|t}-\mu_{k|t,j})^\top\Sigma^{-1}_{k|t,j}(o_{\bar k|t}-\mu_{k|t,j})$}\\
    0 & \text{ otherwise }
    \end{cases} 
\end{align}\normalsize
and $p_{k|t,j}=p_{\bar k|t,j}~\forall k\in\mathcal{I}_{\bar k}^{N-1}$. In summary, we use \eqref{eq:tv_dist_j} and \eqref{eq:belief_update} to obtain the TV prediction model $f^{TV}_k(\cdot)$ in \eqref{opt:TV_dyn} as the multi-modal distribution of $o_{k+1|t}$ conditioned on $o_{k|t}$,
\small
\begin{align}\label{eq:mm_dist_markov}
    o_{k+1|t}|o_{k|t}\sim\sum_{j=1}^Jp_{k|t,j}\mathcal{N}(T_{k|t,j}o_{k|t}+c_{k|t,j},\Tilde{\Sigma}_{k+1|t,j}).
\end{align}
\normalsize
Defining $n_{k|t,j}\sim\mathcal{N}(0,\tilde{\Sigma}_{k+1|t,j})~\forall k\in\mathcal{I}_0^{N-1}, \forall j\in\mathcal{I}_1^J$, we can use properties of Gaussian random variables (closure under linear combinations and the orthogonality principle) to rewrite \eqref{eq:mm_dist_markov} as
\small
\begin{align}\label{eq:TV_dyn_atv}
    o_{k+1|t}=T_{k|t,\sigma_t}o_{k|t}+c_{k|t,\sigma_t}+n_{k,\sigma_t} \ ~\forall k\in\mathcal{I}_0^{N-1}
\end{align}\normalsize
where the posterior distributions of the mode $\sigma_t$ are determined using
\small
\begin{align}\label{eq:mode_dist}
    \sigma_t|o_{k|t}=\begin{cases}\sigma_t & \forall k< \bar k\\
    \sigma_t|o_{\bar k|t} \text{ given by }\eqref{eq:belief_update} & \forall k\geq\bar k
    \end{cases}.
\end{align}\normalsize

\subsection{Multi-modal Collision Avoidance Constraints}
We assume that we are given or can infer the TV rotation matrices for each mode along the prediction horizon as 
\small$\{\{R_{k|t,j}\}_{k=1}^N\}_{j=1}^J.$\normalsize For collision avoidance between the EV and the TV, we impose the following chance constraint
\small
\begin{align}\label{eq:oa_chance_constraint}
    \mathbb{P}(g_{k|t}(P_{k|t},o_{k|t})\geq1\ )\geq 1-\epsilon~~~\forall k\in\mathcal{I}_1^N
\end{align}\normalsize
where $P_{k|t}=[\Delta X_{k|t}\ \Delta Y_{k|t}]^\top + P^{ref}_{k|t}$ with $P^{ref}_{k|t}=[X^{ref}_{t+k}\ Y^{ref}_{t+k}]^\top$ given by \eqref{eq:EV_ref_traj}, $[\Delta X_{k|t}\ \Delta Y_{k|t}]^\top$ given by \eqref{eq:EV_ltv_model}, and 
% \small\begin{align}\label{eq:g_def}
%  &g_{k|t}(P,o)=(P-o)^\top R^\top_{k|t,\sigma_t}\begin{bmatrix}\frac{1}{a_{ca}^2}&0\\0&\frac{1}{b_{ca}^2}\end{bmatrix}R_{k|t,\sigma_t}(P-o).
% \end{align}\normalsize
\small\begin{align}\label{eq:g_def}
 &g_{k|t}(P,o)=\Big\Vert\begin{bmatrix}\frac{1}{a_{ca}}&0\\0&\frac{1}{b_{ca}}\end{bmatrix} R_{k|t,\sigma_t}(P-o)\Big\Vert^2.
\end{align}\normalsize 
% \small\begin{align}\label{eq:g_def}
%  &g_{k|t}(P,o)=(P-o)^\top R^\top_{k|t,\sigma_t}\begin{bmatrix}\frac{1}{a_{ca}^2}&0\\0&\frac{1}{b_{ca}^2}\end{bmatrix}R_{k|t,\sigma_t}(P-o).
% \end{align}\normalsize 
%
$a_{ca}=a_{TV}+d_{EV}, b_{ca}=b_{TV}+d_{EV}$ are semi-axes of the ellipse containing the TV's extent with a buffer of $d_{EV}$. $g_{k|t}(P,o)\geq1$ implies that the EV's extent (modelled as a disc of radius $d_{EV}$ and centre $P$) does not intersect the TV's extent (modelled as an ellipse with semi-axes $a_{TV}, b_{TV}$ and centre $o$, oriented as given by $R_{k|t,\sigma_t}$). 
This constraint is non-convex because it involves the integral of the nonlinear function $g_{k|t}(\cdot)$ over the joint distribution of $(P_{k|t}, o_{k|t}, \sigma_t)$. To address the multi-modality, we conservatively impose the chance constraint conditioned on every mode with the same risk level $1-\epsilon$ because of the following implication
\small
\begin{align*}&\mathbb{P}(g_{k|t,j}(P_{k|t},o_{k|t,j})\geq1 )\geq 1-\epsilon~~\forall j\in\mathcal{I}_1^J\\ \Rightarrow& \mathbb{P}(g_{k|t}(P_{k|t},o_{k|t})\geq1 )\geq 1-\epsilon
\end{align*}
\normalsize
where $g_{k|t,j}(P,o)$ is defined by conditioning \eqref{eq:g_def} on mode $j$. To address the nonlinearity, we use the convexity of $g_{k|t,j}(\cdot)$ to construct its linear under-approximation in the following proposition.

\begin{prop}\label{prop:2}
 For collision avoidance, define EV positions $\forall k\in\mathcal{I}_1^{N},  \forall j\in\mathcal{I}_1^J$ as \small $P^{ca}_{k|t,j}=\mu_{k|t,j}+\frac{1}{\sqrt{g_{k|t,j}(P^{ref}_{k|t},\mu_{k|t,j})}}(P^{ref}_{k|t}-\mu_{k|t,j})$\normalsize. Then for the affine function,
\small$g^L_{k|t,j}(P,o)=\partial_P g_{k|t}(P^{ca}_{k|t,j},\mu_{k|t,j})(P-P^{ca}_{k|t,j})+\partial_o g_{k|t}(P^{ca}_{k|t,j},\mu_{k|t,j})(o-\mu_{k|t,j})$\normalsize, we have $\forall k\in\mathcal{I}_0^{N-1}$
\small
\begin{align}\label{eq:oa_cc_lin}
&\mathbb{P}(g^L_{k+1|t,j}(P_{k+1|t},o_{k+1|t,j})\geq0\ )\geq 1-\epsilon~\forall j\in\mathcal{I}_1^J\\ \Rightarrow& \mathbb{P}(g_{k|t}(P_{k+1|t},o_{k+1|t})\geq1)\geq 1-\epsilon\nonumber
\end{align}
\normalsize
\end{prop}
\begin{proof}
Appendix~\ref{app:prop2proof} in full version \cite{nair2021stochastic}.
\end{proof}
We build on these approximations to complete the reformulation of \eqref{eq:oa_chance_constraint} for EV and TV trajectories given by \eqref{eq:EV_ltv_model}, \eqref{eq:TV_dyn_atv} in closed-loop with a feedback policy.
\subsection{Predictions using EV and TV state feedback policies}
\begin{figure}[H]
    \centering
    \begin{subfigure}{0.45\columnwidth}
        \centering
        \includegraphics[width=\columnwidth]{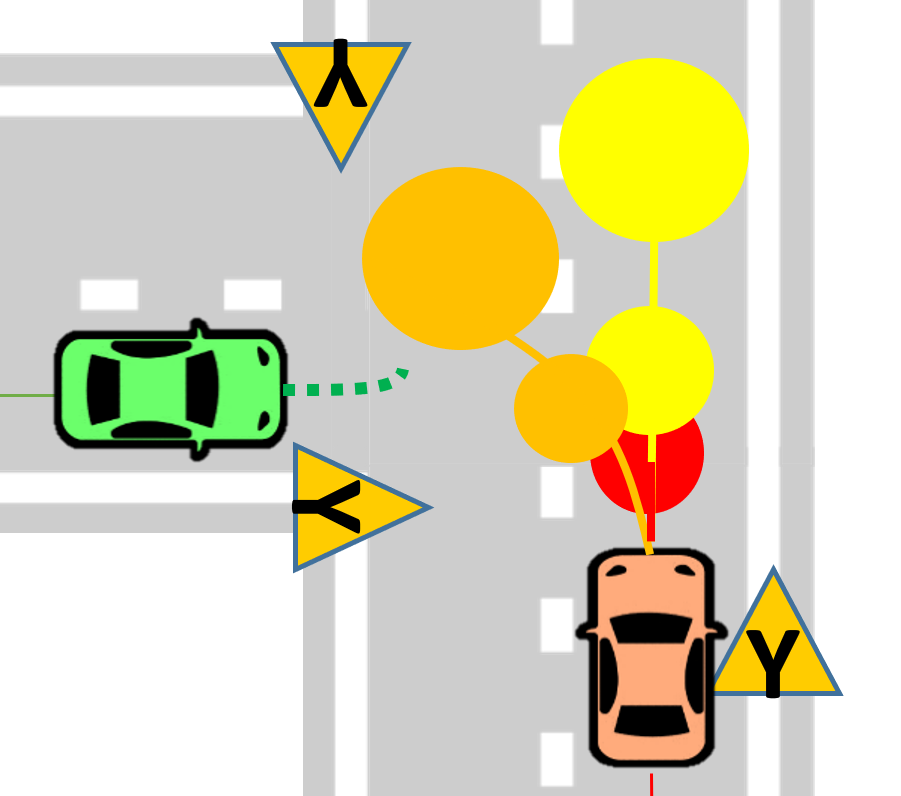} % first figure itself
        \caption{Prediction with open-loop sequences $\mathbf{u}_t\in\mathbb{R}^{2\times N}$}
    \end{subfigure}\hfill %
    \begin{subfigure}{0.45\columnwidth}
        \centering
        \includegraphics[width=\columnwidth]{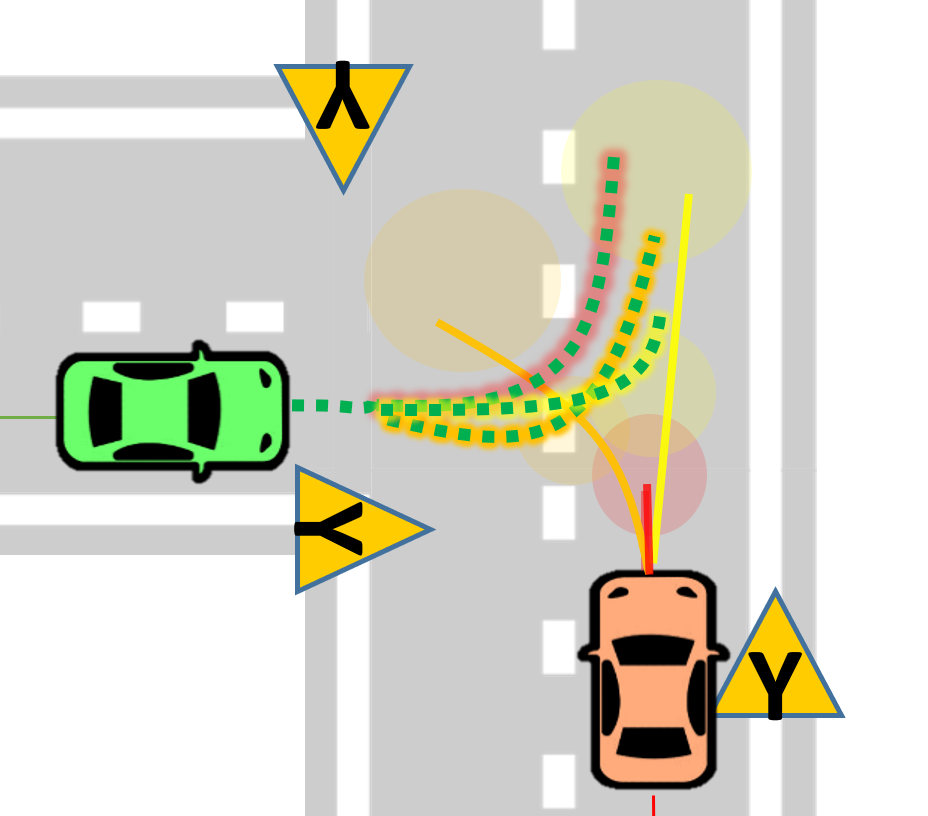} % second figure itself
        \caption{Predictions with policies $\mathbf{u}_t\in\Pi_{\theta_t}(\mathbf{x}_t,\mathbf{o}_t)$}
    \end{subfigure}
    \caption{\small{In (a), solving \eqref{opt:SMPC_skeleton} over open-loop sequences can be conservative because EV prediction (green-dashed) from a single sequence of control inputs must satisfy all the obstacle avoidance constraints. In (b), optimizing over policies \eqref{opt:gen_pol_class} allows for different EV predictions depending on the TV trajectory realizations (green-dashed with highlights corresponding to different TV trajectories).\label{fig:OLvsFB}}}    
\end{figure}
We propose to use parameterized polices $\Pi_{\theta_t}(\mathbf{x}_t,\mathbf{o}_t)$ so that the EV's control $\mathbf{u}_t$ are functions of the EV and TV trajectories $\mathbf{x}_t, \mathbf{o}_t$ along the prediction horizon (as depicted in Fig. \ref{fig:OLvsFB}). Now we proceed to describe the parameterization $\Pi_{\theta_t}(\mathbf{x}_t,\mathbf{o}_t)$ and set of feasible parameters $\Theta_t$. 

Given the EV model \eqref{eq:EV_ltv_model} and mode-dependent TV model \eqref{eq:TV_dyn_atv}, we propose the following feedback policy for the EV:
\small
\begin{align}\label{eq:policy}
    &\Delta u_{k|t}=\pi_{k|t}(x_{k|t}, o_{k|t})=\nonumber\\
    &\begin{cases}h_{k|t}+\sum_{l=0}^{k-1}M_{l,k|t}w_{l|t} + K_{k|t}o_{k|t}&\text{if }k<\bar k\\
    h^1_{k|t}+\sum_{l=0}^{k-1}M^1_{l,k|t}w_{l|t} + K^1_{k|t}o_{k|t}&\text{if }k\geq\bar k,~\sigma_t=1\\
    \vdots &\vdots\\
    h^{J}_{k|t}+\sum_{l=0}^{k-1}M^{J}_{l,k|t}w_{l|t} + K^{J}_{k|t}o_{k|t}&\text{if }k\geq\bar k,~\sigma_t=J
    \end{cases},
\end{align}\normalsize
which uses State Feedback (SF) for the TV states but Affine Disturbance Feedback (ADF) for feedback over EV states instead of SF. As shown in \cite{goulart2006optimization}, the ADF: $h_{k|t}+\sum_{l=0}^{k-1}M_{l,k|t}w_{l|t}$ is equivalent to  SF: $h_{k|t}+F_{k|t}\Delta x_{k|t}$, but the state prediction $\Delta x_{k|t}$ with ADF is affine in $\Delta x_{0|t}$ and the policy parameters up to time $k$: $\{h_{s|t}, \{M_{l,s|t}\}_{l=0}^{s-1}\}_{s=0}^{k}$; whereas with SF,  $\Delta x_{k|t}$ involves nonlinear products of the parameters, $F_{k|t},\dots F_{0|t}$. This is beneficial because \eqref{eq:EV_cc}, \eqref{eq:oa_cc_lin} become affine chance constraints in the ADF policy parameters (at the cost of needing \small$O(N^2)$\normalsize\ parameters instead of \small$O(N)$\normalsize\ for SF). In the proof of Proposition~\ref{prop:3}, we show that despite using SF for the TV states, $\Delta x_{k|t}$ are affine in $\{K_{s|t}\}_{s=0}^k~\forall k\in\mathcal{I}_0^{n-1}$. This is beneficial for scaling our approach to multiple TVs because we use \small$O(2\cdot4\cdot N^2+2\cdot2n_{TV}\cdot N)$\normalsize\ parameters instead of \small${O(2\cdot(4+2n_{TV})\cdot N^2)}$\normalsize\ parameters, where $n_{TV}$ is the number of target vehicles.  Also notice that we use separate mode-dependent policies for $k\geq\bar k$ since the TV mode can be determined by \eqref{eq:mode_dist}. Conditioned on mode $j\in\mathcal{I}_1^J$, let $\mathbf{\Delta u}^j_t=[\Delta u^{\top}_{0|t,j}\dots\Delta u^{\top}_{N-1|t,j}]^\top=\mathbf{h}_t^j+\mathbf{M}_t^{j}\mathbf{w}_t+\mathbf{K}_t^j\mathbf{o}_t$ with $\mathbf{w}_t=[w^\top_{0|t}\dots w^\top_{N-1|t}]^\top$ and the stacked policy parameters $\mathbf{h}_t^j\in\mathbb{R}^{2N}, \mathbf{M}_t^{j}\in\mathbb{R}^{2N\times 4N},  \mathbf{K}_t^j\in\mathbb{R}^{2N\times 2N}$ (\eqref{mat:h}-\eqref{mat:M} in appendix~\ref{app:matrices}\footref{website}). Given $x_t$ and $o_t$, the set of policy parameters that satisfy the EV state and input chance constraints \eqref{eq:EV_cc} and the multi-modal collision avoidance chance constraint \eqref{eq:oa_cc_lin} is given by
\small\begin{align*}
% \label{eq:mixed_ADF_pol_class}
    \Theta_t(x_t,o_t)&=
    \left\{\{\mathbf{h}_t^j,\mathbf{M}_t^j,\mathbf{K}_t^j\}_{j=1}^J\middle\vert\begin{aligned} \eqref{eq:EV_cc}, \eqref{eq:oa_cc_lin} \text{ hold } \forall k\in\mathcal{I}_0^{N-1},\\ \Delta x_{0|t}=x_t-x^{ref}_t, o_{0|t}=o_t\end{aligned}\right\}
\end{align*}\normalsize
and the policy parameterization given parameters $\theta_t=\{\mathbf{h}^j_t,\mathbf{M}_t^j, \mathbf{K}_t^j\}_{j=1}^J\in\Theta_t(x_t,o_t)$ is \small
\begin{align*}\Pi_{\theta_t}(\mathbf{x}_t,\mathbf{o}_t)&=\left\{\mathbf{h}_t^j+\mathbf{M}_t^{j}\mathbf{w}_t+\mathbf{K}_t^j\mathbf{o}_t \middle \vert\begin{aligned} j\in\mathcal{I}_1^J\text{ such that }\\ o_{\bar{k}|t}\in\mathcal{E}(\mu_{\bar{k}|t,j},\Sigma_{\bar{k}|t,j},\beta)\end{aligned} \right\}
\end{align*}
\normalsize
\begin{prop}\label{prop:3} The state predictions $\mathbf{x}_t$ in closed-loop with $\Pi_{\theta_t}(\mathbf{x}_t,\mathbf{o}_t)$ are affine in $\{\mathbf{h}^j_t,\mathbf{M}_t^j, \mathbf{K}_t^j\}_{j=1}^J$.
If $\epsilon<\frac{1}{2}$, then the set of policy parameters $\Theta_t(x_t,o_t)$ is a Second-Order Cone (SOC).
\end{prop}
\begin{proof}
Appendix~\ref{app:prop3proof} in full version \cite{nair2021stochastic}.
\end{proof}

\subsection{SMPC Optimization Problem}
We define the cost \eqref{opt:obj} of the SMPC optimization problem \eqref{opt:SMPC_skeleton} to penalise deviations of the EV state and input trajectories from the reference trajectory \eqref{eq:EV_ref_traj} as
\small
\begin{align*}
% \label{eq:SMPC_cost}
    J_t(\mathbf{x}_t,\mathbf{u}_t)=&\mathbb{E}\left(\sum\limits_{k=0}^{N-1}\Delta x_{k+1|t}^\top Q\Delta x_{k+1|t} + \Delta u_{k|t}^\top R\Delta u_{k|t}\right)
\end{align*}
\normalsize
where $Q\succ 0, R\succ 0$. After explicitly substituting the dynamics constraints \eqref{eq:EV_ltv_model},\eqref{eq:TV_dyn_atv},  we obtain the EV control \eqref{eq:SMPC} from our SMPC formulation as:
%taking into account multi-modal predictions of the TV, compactly written as
\small$$\min\limits_{\{\mathbf{h}_t^j,\mathbf{K}_t^j,\mathbf{M}_t^j\}_{j=1}^J\in \Theta_t(x_t,o_t)}J_t(\mathbf{x}_t,\mathbf{u}_t).$$\normalsize The next proposition characterises this optimization problem.
\begin{prop}\label{prop:4}
The SMPC control action \eqref{eq:SMPC} is given by solving the following Second-Order Cone Program (SOCP):
\small
\begin{equation}\label{opt:SMPC}
	\begin{aligned}
	\min_{\substack{s,\{\mathbf{h}_t^j,\mathbf{K}_t^j,\mathbf{M}_t^j\}_{j=1}^J}}&\  \displaystyle s \\
		\text{s.t. }  \quad\quad&\{\mathbf{h}_t^j,\mathbf{K}_t^j,\mathbf{M}_t^j\}_{j=1}^J\in \Theta_t(x_t,o_t),\\
		 \quad\quad&N_t(\{\mathbf{h}_t^j,\mathbf{K}_t^j,\mathbf{M}_t^j\}_{j=1}^J)\leq s+r_t
	\end{aligned}
\end{equation}\normalsize
where $N_t(\cdot)$ is convex and quadratic in $\{\mathbf{h}_t^j,\mathbf{K}_t^j,\mathbf{M}_t^j\}_{j=1}^J$ and $r_t$ is a constant determined by $x_t$, the EV reference trajectory \eqref{eq:EV_ref_traj} and the process noise covariance $\Sigma_w$. 
\end{prop}
\begin{proof}
Appendix~\ref{app:prop4proof} in full version \cite{nair2021stochastic}
\end{proof}

\section{Experimental Design}
\label{sec:expt_design}
In order to assess the benefits of our proposed stochastic MPC formulation, we use the CARLA~\cite{carla_sim_2017} simulator to run closed-loop simulations\footref{website}.  In this section, we describe the simulation environment, evaluated scenarios, and prediction framework used to generate multimodal predictions.  We then show how the SMPC formulation is integrated into a planning framework to address these predictions.  Finally, we provide a set of metrics and baseline policies used to evaluate the performance of our approach.

\subsection{CARLA Simulation Environment}
\begin{figure}[h]
    \centering
    \includegraphics[width=0.7\columnwidth]{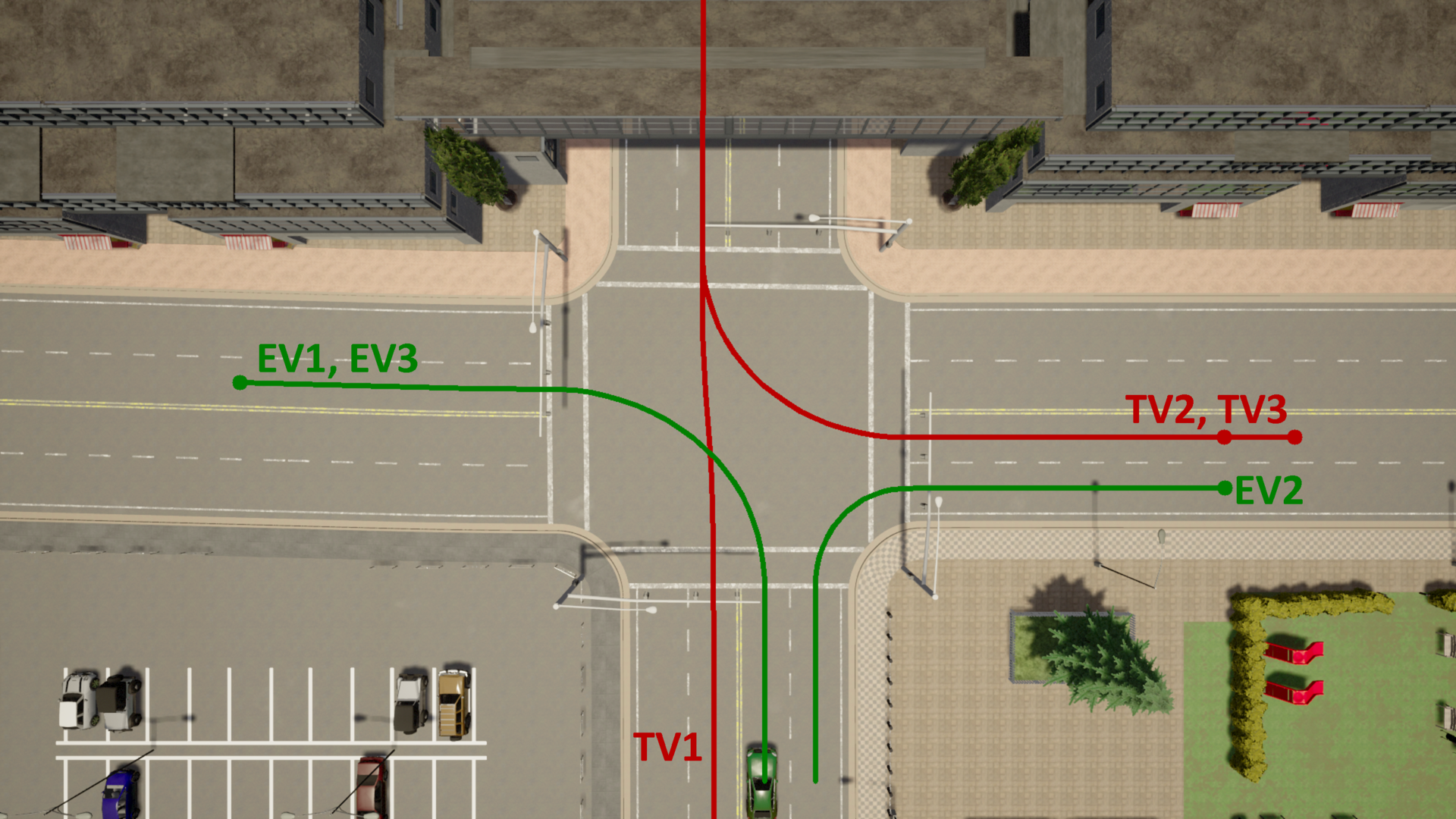}
    \caption{\small{Routes used for EV and TV vehicles with the scenario indicated by the numeric label.  The EV follows a green path, while the TV follows a red path with the same scenario index.}}
    \label{fig:scen_carla}    
\end{figure}
The EV is tasked with passing through the intersection depicted in Figure~\ref{fig:scen_carla}, while avoiding any TVs simultaneously driving through the same intersection.  The scenarios are:
\begin{enumerate}
    \item EV turns left while TV drives straight.
    \item EV turns right while TV turns left into the same road segment.
    \item EV turns left while TV turns left into the opposite road segment.
\end{enumerate}

To ensure improved repeatability of our experiments, we use the synchronous mode of CARLA.  This ensures that all processing (prediction, planning, and control) is complete before advancing the simulation.  Thus, our results only consider the impact of the control inputs selected and not the time delays incurred in computing them\footnote[2]{All experiments were run on a computer with a Intel i9-9900K CPU, 32 GB RAM, and a RTX 2080 Ti GPU.}. 
% MSI Infinite X Plus 9SF-270US Enthusiast Gaming Desktop GeForce RTX 2080 Ti 11G Intel Core i9-9900K 32GB 512GB SSD + 2TB HDD Windows 10 Pro VR Ready + KB and Mouse

The control for the TV is described by a simple nonlinear MPC (NMPC) scheme to track a pre-specified reference. For realism, the TV's NMPC also uses short predictions of the EV for simple distance-based obstacle avoidance constraints.  Additionally, we note that while the intersection has traffic lights, these are ignored in our experiments.  Effectively, each scenario is treated as an interaction at an unsignalized intersection.
%
%The complete set of scenarios and manipulated factors is detailed in Section~\ref{sec:results}.
\small
\begin{table*}[t]
    \centering
    \caption{Closed-loop performance comparison across all scenarios. }
    
    % \todo{Update values}}
    \label{tab:comparison_scenario_all}    
    \begin{tabular}[c]{|c |c | *{1}{c} | *{3}{c} | *{3}{c} | *{1}{c} |}
        %% HEADER
        \hline
        \multirow{2}{*}{\textbf{Scenario}} &
        \multirow{2}{*}{\textbf{Policy}} &
        \multicolumn{1}{c|}{\textbf{Mobility}} &
        \multicolumn{3}{c|}{\textbf{Comfort}} &
        \multicolumn{3}{c|}{\textbf{Conservatism}} &
        \multicolumn{1}{c|}{\textbf{Efficiency}} \\
        & & \mobility & \comfort & \conservatism & \efficiency \\
        \hline
        %% ENTRIES
        \multirow{3}{*}{\textbf{1}}
        & \textbf{SMPC\_BL}    & 1.099 & 1.830 & 7.336 & 12.824 & 5.396 & 5.270 &  98.66 & 33.063 \\
        & \textbf{SMPC\_OL}    & $\mathbf{1.096}$ & 2.113 & $\mathbf{2.072}$ & 10.746 & 3.790 & 4.848 & $\mathbf{100.00}$ &  $\mathbf{1.600}$ \\
        & \textbf{SMPC}        & 1.232 & $\mathbf{1.217}$ & $\mathbf{2.098}$ &  $\mathbf{6.554}$ & $\mathbf{1.242}$ & $\mathbf{3.578}$ & $\mathbf{100.00}$ & 62.399 \\
        \hline
        \multirow{3}{*}{\textbf{2}}
        & \textbf{SMPC\_BL}    & 1.366 & 1.712 & 12.334 & 6.016 & 1.130 & 10.174 &  78.59 & 44.280 \\
        & \textbf{SMPC\_OL}    & 1.571 & 1.553 &  7.929 & $\mathbf{5.557}$ & 3.900 &  9.288 &  78.01 &  $\mathbf{1.676}$ \\
        & \textbf{SMPC}        & $\mathbf{1.294}$ & $\mathbf{1.012}$ &  $\mathbf{2.531}$ & 5.691 & $\mathbf{1.027}$ &  $\mathbf{7.884}$ & $\mathbf{100.00}$ & 63.949 \\
        \hline
        \multirow{3}{*}{\textbf{3}}
        & \textbf{SMPC\_BL}    & 1.370 & 1.620 & 10.971 & $\mathbf{4.378}$ & 2.707 & 11.582 &  83.09  & 43.372 \\
        & \textbf{SMPC\_OL}    & 1.211 & 1.878 &  3.944 & 8.247 & 1.603 & 11.962 &  93.07  &  $\mathbf{1.641}$ \\
        & \textbf{SMPC}        & $\mathbf{1.152}$ & $\mathbf{1.159}$ &  $\mathbf{2.143}$ & 8.595 & $\mathbf{0.734}$ &  $\mathbf{7.917}$ & $\mathbf{100.00}$  & 61.400 \\
        %% END
        \hline
    \end{tabular}
\end{table*}
\normalsize
\subsection{Multimodal Prediction Architecture}
A key assumption in our SMPC approach is that the TV's future motion is described by a GMM \eqref{eq:GMM}.  To generate  this probability distribution, we implement a variant of the MultiPath prediction model proposed in \cite{multipath_2019}.  The input to this neural network architecture is a semantic image and pose history for the target vehicle in the scene.  These inputs are used to generate a GMM over future trajectories by regression with respect to a set of predetermined anchors and classification over the anchor probabilities.  We note, however, that our SMPC approach is appropriate for any prediction model that produces a GMM trajectory distribution.

To train the model, we selected a subset of the Lyft Level 5 prediction dataset~\cite{l5kit2020} with $16$ anchor trajectories identified using k-means clustering.  At runtime, we opted to truncate the GMM produced by MultiPath.  In particular, we pick the top $J=3$ modes and renormalize the mode probabilities correspondingly.

\subsection{Motion Planning Architecture}
Given the predicted multimodal distribution for the TV, we use the framework introduced in Section~\ref{sec:SMPC} to generate feedback control policies for deployment.  
In addition to predictions of the TV, a dynamically feasible EV reference trajectory \eqref{eq:EV_ref_traj} must be provided. In our work, a NMPC problem is solved using IPOPT~\cite{ipopt_2006} which tracks a high-level route provided by the CARLA waypoint API, while incorporating dynamic and actuation constraints.
Given the EV reference \eqref{eq:EV_ref_traj} and TV predictions \eqref{eq:GMM}, the SMPC optimization problem \eqref{opt:SMPC} is solved using Gurobi \cite{gurobi} to compute feedback control \eqref{eq:SMPC}. When an infeasible problem is encountered, a braking control is commanded. The corresponding control action $u^{*}_{0|t}$ is given as a reference to a low-level control module that sets the vehicle's steering, throttle, and brake inputs.  These loops are repeated until all vehicles in the scenario reach their destination.

\subsection{Policies}\label{sec:pol_list}
We evaluate the following set of policies: %introduce a set of baseline SMPC policies for comparison.
\begin{itemize}
    \item \textbf{SMPC}: Our proposed framework, given by solving \eqref{opt:SMPC} which optimizes over our policies \eqref{eq:policy}.
    \item \textbf{SMPC\_OL}: This model is an ablation of our approach, where $\mathbf{u}_t\in\mathbb{R}^{2\times N}$ replaces \eqref{opt:gen_pol_class} and the GMM~\eqref{eq:GMM} is used directly in \eqref{opt:TV_dyn} instead of \eqref{eq:mm_dist_markov}.
    \item \textbf{SMPC\_BL}: Based off \cite{wang2020non}, a nonlinear SMPC algorithm is used with the collision avoidance chance constraints \eqref{opt:oa_constr}  reformulated using VP concentration bounds.  The optimization is performed over open-loop sequences, i.e., $\mathbf{u}_t\in\mathbb{R}^{2\times N}$ instead of \eqref{opt:gen_pol_class}, and no additional structure is assumed as in \eqref{opt:TV_dyn}, i.e., $o_{k+1|t}\sim f^{TV}_k(o_t)$.
\end{itemize}

\subsection{Evaluation Metrics}
We introduce a set of closed-loop behavior metrics to evaluate the policies.  A desirable planning framework enables high mobility without being over-conservative, allowing the timely completion of the driving task, while maintaining passenger comfort. The computation time should  also not be exorbitant to allow for real-time processing of updated scene information.
The following metrics are used to assess these factors:
\begin{itemize}
    \item \textbf{Mobility}: (1) $\tilde{\mathcal{T}}_{episode}$: Time the EV takes to reach its goal normalised by the time taken by the reference.
    \item \textbf{Comfort}: (1) $\tilde{\mathcal{A}}_{lat}$: Peak lateral acceleration normalised by the peak lateral acceleration in reference, (2) $\bar{\mathcal{J}}_{long}$: Average longitudinal jerk,  and (3) $\bar{\mathcal{J}}_{lat}$: Average lateral jerk.  High values are undesirable, linked to sudden braking or steering.
    \item \textbf{Conservatism}: (1) $\Delta\tau$: Deviation of the closed-loop trajectory from the reference trajectory, (2) $\bar{d}_{min}$: Closest distance between the EV and TV, and (3) $\mathcal{F}$: Feasibility \% of the SMPC optimization problem.  While an increase in (1) and (2) corresponds to higher safety, a large value could indicate an over-conservative planner and result in low values of (3).
    \item \textbf{Computational Efficiency}: (1) $\bar{\mathcal{T}}_{solve}$: Average time taken by the solver; lower is better.
\end{itemize}

\section{Results}
\label{sec:results}
%In this section, we describe the specific scenarios simulated in CARLA for deploying the various SMPC policies (\textbf{SMPC}, \textbf{SMPC\_OL} and \textbf{SMPC\_BL}). The different scenarios are characterised by the desired routes of the EV and TV. For each scenario, we deploy the policies for four different initial conditions and tabulate the metrics as described in Section~\ref{sec:expt_design}.
Now, we present the results of the various SMPC policies (Section~\ref{sec:pol_list}). For each scenario, we deploy each policy for 4 different initial conditions by varying: (1) initial distance from the intersection$\in\{10 \text{ m}, 20 \text{ m}\}$ and (2) initial speed$\in\{8 \text{m s}^{-1}, 10 \text{m s}^{-1}\}$. For all the policies, we use a prediction horizon of $N=10$, a discretization time-step of $dt= 0.2 \text{ s}$, and a risk level of $\epsilon=0.05$ for the chance constraints.

% \begin{itemize}
%     \item episodes inits
%     \item MPC parameters (N, risk params, dt). Constraints parameters deferred to code
%     \item REference data table
%     \item for scenario \todo{x}, compare policy snapshots.  explain verbally how feedback is helping.
% \end{itemize}

%For each scenario, we deploy each SMPC policy for 4 episodes. Each episode is initialised with a different EV initial condition by choosing: (1) the starting distance from the intersection, and (2) the initial speed. In all cases, the TV starts in the oncoming lane (i.e. in the opposite direction as the EV).

The performance metrics, averaged across the initial conditions, are shown in Table~\ref{tab:comparison_scenario_all}.  In terms of mobility, \textbf{SMPC} is able to improve or maintain mobility compared to the baselines.  There is a noticeable improvement in comfort and conservatism metrics, as the \textbf{SMPC} can stay close to the TV-free reference trajectory without incurring high acceleration/jerk and keeping a safe distance from the TV. Remarkably, \textbf{SMPC} was also able to find feasible solutions for the SMPC optimization problem 100\% of the time in our experiments, because the formulation optimizes over policies. Finally, while the \textbf{SMPC\_OL} ablation is the fastest in solve time (due to fewer decision variables), the \textbf{SMPC} can be used real-time at approximately 20 Hz with further optimization.  

The results highlight the benefits of optimizing over parameterized policies in the SMPC formulation for the EV, towards collision avoidance given the multi-modal GMM predictions of the TV.

\bibliographystyle{IEEEtran}
\bibliography{references.bib}

\appendix
\subsection{Matrix definitions}\label{app:matrices}
  \footnotesize 
\begin{align}
    &\mathbf{h}_t^j=[h^\top_{0|t}\dots h^\top_{\bar k-1|t}\ h^{j\top}_{\bar k|t}\dots h^{j\top}_{N-1|t} ]^\top\label{mat:h}\\
    &\mathbf{K}_t^j=\text{blkdiag}\left(K_{0|t},\dots,K_{\bar k-1|t}, K^j_{\bar k|t},\dots, K^j_{N-1|t}\right)\label{mat:K}\\
    &\mathbf{M}_t^{j}=\begin{bmatrix}
    O&\hdots&\hdots&\hdots& O\\M_{0,1|t}&O&\hdots&\hdots& O\\
    \vdots&\vdots&\vdots&\vdots\\
    M_{0,\bar k-1|t}&\hdots M_{\bar k-2,\bar k-1|t}&O&\hdots& O\\
    M^{j}_{0,\bar k|t}&\hdots &M^{j}_{\bar k-1,\bar k|t}&\hdots& O\\
    \vdots&\vdots&\vdots&\vdots\\
    M^{j}_{0,N-1|t}&\hdots &\hdots& M^{j}_{N-2,N-1|t}&O
    \end{bmatrix}\label{mat:M}\\
    &\mathbf{A}_t=\begin{bmatrix}I_4\\ A_{0|t}\\A_{1|t}A_{0|t}\\\vdots\\\prod\limits_{k=0}^{N-1}A_{k|t}\end{bmatrix},  \mathbf{B}_t=\begin{bmatrix}O&\hdots&\hdots& O\\B_{0|t}&O&\hdots&O\\A_{1|t}B_{0|t}&B_{1|t}&\hdots&O\\\vdots&\ddots&\ddots&\vdots\\\prod\limits_{k=1}^{N-1}A_{k|t}B_{0|t}&\hdots&\dots&B_{N-1|t}\end{bmatrix},\label{mat:AB}\\
    &\mathbf{T}^j_t=\begin{bmatrix}I_2\\ T_{0|t,j}\\T_{1|t,j}T_{0|t,j}\\\vdots\\\prod\limits_{k=0}^{N-1}T_{k|t,j}\end{bmatrix}, \mathbf{C}^j_t=\begin{bmatrix}O\\c_{0|t,j}\\c_{1|t,j}+T_{1|t,j}c_{0|t,j}\\\vdots\\c_{N-1|t,j}+\sum\limits_{k=0}^{N-1}\prod\limits_{l=k+1}^{N-1}T_{l|t,j} c_{k|t,j}\end{bmatrix}\label{mat:TC}\\
    &\mathbf{E}_{t}=\begin{bmatrix}O&\hdots&\hdots& O\\I_4&O&\hdots&O\\A_{1|t}&I_4&\hdots&O\\\vdots&\ddots&\ddots&\vdots\\\prod\limits_{k=1}^{N-1}A_{k|t}&\hdots&\dots&I_4\end{bmatrix},\  \mathbf{L}^j_{t}=\begin{bmatrix}O&\hdots&\hdots& O\\I_2&O&\hdots&O\\T_{1|t,j}&I_2&\hdots&O\\\vdots&\ddots&\ddots&\vdots\\\prod\limits_{k=1}^{N-1}T_{k|t,j}&\hdots&\dots&I_2\end{bmatrix} \label{mat:EL}\\
    &\boldsymbol{\Sigma}_w=I_N\otimes \Sigma_w,\  \boldsymbol{\Sigma}^j_n=\text{blkdiag}(I_{N-1}\otimes \Sigma_n,\Sigma_{1|t,j}) \label{mat:Sigma}\\
    &\mathbf{Q}=I_{N+1}\otimes Q,\ \mathbf{R}=I_N\otimes R,\label{mat:QR}
    \end{align}
\normalsize 
\subsection{Proof of proposition~\ref{prop:1}}\label{app:prop1proof}
Defining $n_{k|t,j}\sim\mathcal{N}(0,\bar\Sigma_{k+1|t,j})~\forall k\in\mathcal{I}_0^{N-1}$, we can use the closure of Gaussian random variables under linear combinations and the orthogonality principle, to rewrite \eqref{eq:tv_dist_j} as
$o_{k+1|t,j}=T_{k|t,j}o_{k|t,j}+c_{k|t,j}+n_{k|t,j}$.
\begin{itemize}
\item Using the definition of $c_{k|t,j}$, we have $o_{k+1|t,j}=T_{k|t,j}(o_{k|t,j}-\mu_{k|t,j})+\mu_{k+1|t,j}+n_{k|t,j}$. Thus, 
\small
\begingroup
\allowdisplaybreaks
\begin{align*}
\mathbb{E}&[\Vert o_{k+1|t,j}-\mu_{k+1|t,j}\Vert^2_{\Sigma^{-1}_{k+1|t,j}}|\ o_{k|t,j}]\\
=&\mathbb{E}[\Vert T_{k|t,j}(o_{k|t,j}-\mu_{k|t,j})+n_{k|t,j}\Vert^2_{\Sigma^{-1}_{k+1|t,j}}|\ o_{k|t,j}]\\
=&\mathbb{E}[\Vert\sqrt{\Sigma^{-1}_{k+1|t,j}}(T_{k|t,j}(o_{k|t,j}-\mu_{k|t,j})+n_{k|t,j})\Vert^2|\ o_{k|t,j}]\\
=&\Vert\sqrt{\Sigma^{-1}_{k+1|t,j}}T_{k|t,j}(o_{k|t,j}-\mu_{k|t,j})\Vert^2\\&~~~~
+\mathbb{E}[\Vert\sqrt{\Sigma^{-1}_{k+1|t,j}} n_{k|t,j}\Vert^2]\\
=&\Vert\sqrt{\Sigma^{-1}_{k|t,j}}(o_{k|t,j}-\mu_{k|t,j})\Vert^2+\text{tr}(\Sigma^{-1}_{k+1|t,j}\Sigma_n)\\
=&S^2+\text{tr}(\Sigma^{-1}_{k+1|t,j}\Sigma_n)
\end{align*}
\endgroup
\normalsize where the second to last equality comes from the definition of $T_{k|t,j}$.
    \item Expressing $o_{k|t,j}$ in terms of $o_t$, we get
\small
\begin{align*}
o_{k|t,j}=&\prod\limits_{l=0}^{k-1}T_{l|t,j}o_{t}+\sum\limits_{m=0}^{k-1}\prod\limits_{l=m}^{k-2}T_{l+1|t,j}(c_{m|t,j}+n_{m|t,j})\\
=&\sqrt{\Sigma_{k|t,j}}\sqrt{\Sigma_{1|t,j}^{-1}}o_t\\
~&+\mu_{k|t,j}-\sqrt{\Sigma_{k|t,j}}\sqrt{\Sigma_{k-1|t,j}^{-1}}\mu_{k-1|t,j}\\
~&+\sqrt{\Sigma_{k|t,j}}(\sqrt{\Sigma_{k-1|t,j}^{-1}}\mu_{k-1|t,j}-\sqrt{\Sigma_{k-2|t,j}^{-1}}\mu_{k-2|t,j})\\
~&\dots-\sqrt{\Sigma_{k|t,j}}\sqrt{\Sigma_{1|t,j}^{-1}}o_t\\
~&+n_{k-1|t,j}+\sum\limits_{m=0}^{k-2}\prod\limits_{l=m}^{k-2}T_{l+1|t,j}n_{m|t,j}\\
=&\mu_{k|t,j}+n_{k-1|t,j}+\sqrt{\Sigma_{k|t,j}}\sqrt{\Sigma_{1|t,j}^{-1}}n_{0|t,j}\\&+\sum\limits_{m=1}^{k-2}\prod\limits_{l=m}^{k-2}\sqrt{\Sigma_{k-1|t,j}}\sqrt{\Sigma_{m|t,j}^{-1}}n_{m|t,j}
\end{align*}
\normalsize
Computing the mean and expectation of this random variable gives the desired distribution. Since the ellipsoids \small$\mathcal{E}(\mu_{k|t,j},\Sigma_{k|t,j}+\Sigma_n+\bar{\Sigma}_{k|t,j}, \beta), \mathcal{E}(\mu_{k|t,j},\Sigma_{k|t,j}, \beta)$\normalsize\ share the same center, the desired containment is verified by seeing that $\beta(\Sigma_{k|t,j}+\Sigma_n+\bar{\Sigma}_{k|t,j})\succ\beta\Sigma_{k|t,j}$.
\end{itemize}
\subsection{Proof of proposition~\ref{prop:2}}\label{app:prop2proof}
 First note that $P^{ca}_{k|t,j}$ is defined such that $g_{k|t,j}(P^{ca}_{k|t,j}, \mu_{k|t,j})=1$. For any convex function $f(x)$, we have $\forall x_0,x: f(x)\geq f(x_0)+\partial_x f(x_0)(x-x_0)$. Since we $g_{k|t,j}(\cdot)$ is convex, we have
 \small $g_{k|t,j}(P,o)\geq g_{k|t,j}(P^{ca}_{k|t,j}, \mu_{k|t,j})+g^L_{k|t,j}(P,o)=1+g^L_{k|t,j}(P,o)$\normalsize.\ So $g^L_{k|t,j}(P,o)\geq 0\Rightarrow g_{k|t,j}(P,o)\geq 1$, and $\forall k\in\mathcal{I}_1^N$,
 \small
 \begin{align*}
 &\mathbb{P}(g^L_{k|t,j}(P_{k|t},o_{k|t,j})\geq0\ )\geq 1-\epsilon~\forall j\in\mathcal{I}_1^J\\ \Rightarrow& \mathbb{P}(g_{k|t,j}(P_{k|t},o_{k|t,j})\geq1)\geq 1-\epsilon ~\forall j\in\mathcal{I}_1^J
 \\ \Rightarrow& \mathbb{P}(g_{k|t}(P_{k|t},o_{k|t})\geq1)\geq 1-\epsilon
 \end{align*}\normalsize
 $\hfill\blacksquare$
  \subsection{Proof of proposition~\ref{prop:3}}\label{app:prop3proof}
  Let $\boldsymbol{\Delta x}_t=[\Delta x^\top_{0|t}\dots \Delta x^\top_{N|t}]^\top$, $\boldsymbol{\Delta u}_t=[\Delta u^\top_{0|t}\dots \Delta u^\top_{N-1|t}]^\top$,  $\mathbf{o}^j_t=[o^\top_{0|t}\ o^\top_{1|t,j}\dots o^\top_{N|t,j}]^\top$ and $\mathbf{n}^j_t=[n^\top_{0|t,j}\dots n^\top_{N-1|t,j}]^\top$. Then the random variable $\boldsymbol{\Delta x}_t$ for the EV trajectory with policy \eqref{eq:policy} conditioned on the mode $\sigma_t=j$, given by $\boldsymbol{\Delta u}^j_t=\mathbf{M}^j_t\mathbf{w}_t+\mathbf{h}^j_t+\mathbf{K}^j_t\mathbf{o}^j_t$, can be written as
  \small
  \begin{align*}
      \boldsymbol{\Delta x}_t=&\mathbf{A}_t\Delta x_{0|t}+\mathbf{B}_t\mathbf{h}^j_t+\mathbf{B}_t\mathbf{K}^j_t(\mathbf{T}^j_to_t+\mathbf{C}^j_t+\mathbf{L}^j_t\mathbf{n}^j_t)\\+&(\mathbf{E}_t+\mathbf{B}_t\mathbf{M}^j_t)\mathbf{w}_t.
  \end{align*}
  \normalsize
  where the matrices $\mathbf{A}_t, \mathbf{B}_t,\mathbf{T}^j_t, \mathbf{C}_t^j, \mathbf{L}_t^j, \mathbf{E}_t$ are defined in \eqref{mat:AB}-\eqref{mat:EL}. Clearly, $\boldsymbol{\Delta x}_t$ (and so, $\mathbf{x}_t$ too) is affine in $\mathbf{h}^j_t,\mathbf{M}^j_t,\mathbf{K}^j_t$. 
  Let the constant matrices $S^x_k, S^o_k, S^u_k$ be such that $S^x_k\boldsymbol{\Delta x}_t=\Delta x_{k|t}$, $S^o_k\mathbf{o}^j_t=o_{k|t,j}$, $S^u_k\boldsymbol{\Delta u}_t=\Delta u_{k|t}$. Defining \footnotesize$G^x_{k|t,j}=\partial_P g_{k|t,j}(P^{ca}_{k|t,j},\mu_{k|t,j}),\ G^o_{k|t,j}=\partial_o g_{k|t,j}(P^{ca}_{k|t,j},\mu_{k|t,j}),\  \tilde{g}_{k|t,j}=G^x_{k|t,j}(P^{ca}_{k|t}-P^{ref}_{k|t,j})+G^o_{k|t,j}\mu_{k|t,j}$\normalsize,\ we can rewrite the affine chance constraint \eqref{eq:oa_cc_lin} for each $j\in\mathcal{I}_1^J$, $k\in\mathcal{I}_1^N$ as
  \footnotesize
  \begin{align*}
      &\mathbb{P}(G^x_{k|t,j}S^x_k\boldsymbol{\Delta x}_t+G^o_{k|t,j}S^o_k\mathbf{o}^j_t\geq \tilde g_{k|t,j})\geq 1-\epsilon\\
      \Rightarrow&\mathbb{P}\Big([G^x_{k|t,j}S^x_k(\mathbf{E}_t+\mathbf{B}_t\mathbf{M}^j_t)\ (G^o_{k|t,j}S^o_k+G^x_{k|t,j}S^x_k\mathbf{B}_t\mathbf{K}^j_t)\mathbf{L}^j_t]\begin{bmatrix}\mathbf{w}_t\\\mathbf{n}^j_t\end{bmatrix}\\
      &\leq \tilde g_{k|t,j}-(G^o_{k|t,j}S^o_k+G^x_{k|t,j}S^x_k\mathbf{B}_t\mathbf{K}^j_t)(\mathbf{T}^j_to_t+\mathbf{C}^j_t)\\
      &-G^x_{k|t,j}S^x_k(\mathbf{A}_t\Delta x_{0|t}+\mathbf{B}_t\mathbf{h}^j_t)\Big)\leq\epsilon\\
      \Rightarrow& G^x_{k|t,j}S^x_k(\mathbf{A}_t\Delta x_{0|t}+\mathbf{B}_t\mathbf{h}^j_t)+(G^o_{k|t,j}S^o_k+G^x_{k|t,j}S^x_k\mathbf{B}_t\mathbf{K}^j_t)(\mathbf{T}^j_to_t+\mathbf{C}^j_t)\\
      +&\Vert\left[G^x_{k|t,j}S^x_k(\mathbf{E}_t+\mathbf{B}_t\mathbf{M}^j_t)\ (G^o_{k|t,j}S^o_k+G^x_{k|t,j}S^x_k\mathbf{B}_t\mathbf{K}^j_t)\mathbf{L}^j_t\right]\boldsymbol{\Sigma}^{j\frac{1}{2}} \Vert_2\boldsymbol{\Phi}^{-1}(\epsilon)\\
      &\geq\tilde{g}_{k|t,j}
  \end{align*}
  \normalsize
  where $\boldsymbol{\Phi}^{-1}(\cdot)$ is the inverse CDF of $\mathcal{N}(0,1)$,\footnotesize$\ 
  \boldsymbol{\Sigma}^j=\text{blkdiag}(\boldsymbol{\Sigma}_w, \boldsymbol{\Sigma}^j_n)
  $\normalsize\ using \eqref{mat:Sigma}. Since $\epsilon<\frac{1}{2}$, we have $\boldsymbol{\Phi}^{-1}(\epsilon)<0$ and the above inequality can be shown to be of the form $\Vert Ax+b\Vert_2\leq c^\top x+d$, a Second-Order Cone (SOC) constraint in $x=(\mathbf{h}^j_t, \mathbf{M}^j_t, \mathbf{K}^j_t)$. The intersection of these constraints $\forall j\in\mathcal{I}_1^J, \forall k\in\mathcal{I}_1^N$ is also a SOC set in $\{\mathbf{h}^j_t,\mathbf{M}^j_t,\mathbf{K}^j_t\}_{j=1}^J$.
  Similarly the affine state and input chance constraints \eqref{eq:EV_cc} can be reformulated into SOC constraints, giving us the desired result. $\hfill\blacksquare$
    \subsection{Proof of proposition~\ref{prop:4}}\label{app:prop4proof}
    We have to show that \small$\min\{J_t(\mathbf{x}_t,\mathbf{u}_t)\ :\{\mathbf{h}^j_t,\mathbf{M}^j_t,\mathbf{K}^j_t\}_{j=1}^J\in\Theta_t(x_t,o_t)\}$ \normalsize is equivalent to \eqref{opt:SMPC} and is a SOCP. We have already shown that $\{\mathbf{h}^j_t,\mathbf{M}^j_t,\mathbf{K}^j_t\}_{j=1}^J\in\Theta_t(x_t,o_t)$ is a SOC constraint in the previous proposition. It remains to show that the cost can be reformulated appropriately. For this, we transform the problem into its epigraph form by introducing a additional scalar decision variable $s$, to get \small$\min\{s:J(x_t,o_t)\leq s, \{\mathbf{h}^j_t,\mathbf{M}^j_t,\mathbf{K}^j_t\}_{j=1}^J\in\Pi(x_t,o_t) \}$\normalsize. Then,
    \footnotesize
    \begingroup
    \allowdisplaybreaks
    \begin{align*}
        &s\geq J_t(\mathbf{x}_t,\mathbf{u}_t)\\
        &=\sum\limits_{j=1}^Jp_{t,j}\mathbb{E}\left(\sum\limits_{k=0}^{N-1}\Delta x_{k+1|t}^\top Q\Delta x_{k+1|t} + \Delta u_{k|t}^\top R\Delta u_{k|t}| \sigma_t=j\right)\\
        &=-\Delta x^\top_{0|t}Q\Delta x_{0|t}+\sum_{j=1}^Jp_{t,j}\mathbb{E}(\boldsymbol{\Delta x}^\top_t\mathbf{Q}\boldsymbol{\Delta x}_t+\boldsymbol{\Delta u}^{j\top}_t\mathbf{R}\boldsymbol{\Delta u}^j_t| \sigma_t=j)\\
        &=\sum_{j=1}^Jp_{t,j}((\mathbf{h}_t^j+\mathbf{K}^j_t(\mathbf{T}_t^jo_t+\mathbf{C}_t^j))^\top(\mathbf{B}^\top_t\mathbf{Q}\mathbf{B}_t+\mathbf{R})(\mathbf{h}_t^j+\mathbf{K}^j_t(\mathbf{T}_t^jo_t+\mathbf{C}_t^j))\\
		&+\text{tr}\big(\boldsymbol{\Sigma}_w\mathbf{M}_t^{j\top}(\mathbf{B}_t^\top\mathbf{Q}\mathbf{B}_t+\mathbf{R})\mathbf{M}_t^j+\mathbf{L}^{j}_t\boldsymbol{\Sigma}_n\mathbf{L}^{j\top}_t\mathbf{K}_t^{j\top}(\mathbf{B}_t^\top\mathbf{Q}\mathbf{B}_t+\mathbf{R})\mathbf{K}_t^j\\
		&+2\Delta x^\top_{0|t}\mathbf{A}^\top_t\mathbf{Q}\mathbf{B}_t(\mathbf{h}_t^j+\mathbf{K}^j_t(\mathbf{T}_t^jo_t+\mathbf{C}_t^j))\big))\\
		&+\underbrace{\Delta x^\top_{0|t}(\mathbf{A}^\top_t\mathbf{Q}\mathbf{A}_t-Q)\Delta x_{0|t}+\text{tr}(\boldsymbol{\Sigma}_w\mathbf{E}^\top_t\mathbf{Q}\mathbf{E}_t)}_{-r_t}
    \end{align*}
    \endgroup
    \normalsize
    and the first three terms correspond to $N_t(\cdot)$, which is quadratic and convex because $Q\succ 0, R\succ 0$. Thus this inequality is a convex quadratic constraint in $(s,\{\mathbf{h}^j_t,\mathbf{M}^j_t,\mathbf{K}^j_t\}_{j=1}^J)$, which is a special case of SOC. Since the cost $s$ is linear, we have that \eqref{opt:SMPC} is a SOCP. $\hfill\blacksquare$

\end{document}